\documentclass[11pt]{article}
\usepackage{amsmath,amsthm}
\usepackage{fullpage,palatino,setspace}
\usepackage{comment}
\usepackage{graphicx}
\usepackage{enumitem}
\usepackage{todonotes}
\usepackage{subfigure}
\usepackage[ruled,vlined]{algorithm2e}
\usepackage{amssymb,amsfonts}
\usepackage[colorlinks=true, linkcolor=blue, citecolor=blue, urlcolor=blue]{hyperref}
\usepackage[capitalise]{cleveref}     
\usepackage{tikz}
\usepackage{complexity}
\usepackage[all]{xy}
\usetikzlibrary{backgrounds}

\newtheorem{theorem}{Theorem}[section]
\newtheorem{lemma}[theorem]{Lemma}

\newtheorem{claim}[theorem]{Claim}
\newtheorem{corollary}[theorem]{Corollary}

\newcommand{\ACz}{\mathsf{AC^0}}

\theoremstyle{remark}

\theoremstyle{definition}

\usepackage{thmtools}
\usepackage{thm-restate}
\usepackage{enumitem}
\usepackage[numbers,sort&compress]{natbib} 

\newcommand*{\ldblbrace}{\{\mskip-5mu\{}
\newcommand*{\rdblbrace}{\}\mskip-5mu\}}

\newcommand{\MCVP}{{\sf MCVP}}

\newcommand{\ktintest}{\texorpdfstring{$(k + 1)$-\textnormal{\textsc{Tinhofer Testing}}}{(k+1)-Tinhofer Testing}\xspace}

\title{A Hierarchy of Tinhofer Graphs: Separations and Membership Testing}
\author{Sutanay Bhattacharjee\thanks{Indian Institute of Technology Madras, Chennai, India.} \and Ameya Panse\thanks{Part of this work was done when the second author was a masters student at IIT Madras}
\and Jayalal Sarma$^*$}
\date{}

\usepackage{collect}

\def\movetoappendix{1}

\definecollection{appendix}
\makeatletter
\newenvironment{aproof}[2]
  { \@nameuse{collect}{appendix}
  { \subsection{#1} \label{#2} \begin{proof} } {\end{proof}}
  }{\@nameuse{endcollect}}
\makeatother

\makeatletter
\newenvironment{appsection}[2]
  { \@nameuse{collect}{appendix}
  { \subsection{#1} \label{#2} }
  {}
  }{\@nameuse{endcollect}}
\makeatother

\ifthenelse{\equal{\movetoappendix}{0}}{

}{}

\usepackage{todonotes}
\newcounter{todocounter}

\begin{document}
\maketitle
\begin{abstract}
Color refinement is an important technique that works very well in practice for the graph isomorphism problem. Tinhofer graphs are the class of graphs for which refinement together with individualization correctly tests graph isomorphism against every other graph, irrespective of the choices of vertices made during individualization. Motivated by the fact that Tinhofer graphs form a natural boundary for efficient graph isomorphism tests based on color refinement, in this paper, we introduce a hierarchy of graph classes within the class of Tinhofer graphs. We call a graph $G$ \emph{$k$-Tinhofer} if, after $k$ rounds of individualization and refinement, the resulting colored graphs remain isomorphic for every graph $H \cong G$, irrespective of the choices of vertices made during individualization.

Arvind {\em et al.}~(2017) studied a hierarchy of graph classes motivated by color refinement - discrete, amenable, Tinhofer, and refinable graphs. We show that the $k$-Tinhofer hierarchy lies between the class of all graphs and Tinhofer graphs, with refinable graphs coinciding with the first level of the hierarchy. We establish the following structural and algorithmic results about this hierarchy:
\begin{itemize}
    \item We obtain two characterizations of $k$-Tinhofer graphs: an algebraic characterization in terms of orbit partitions induced by pointwise stabilizers of automorphism groups, and a combinatorial characterization in terms of individualization-refinement trees and quotient graphs. 
    \item For every fixed integer $k \ge 0$, there exist vertex-colored graphs that are $k$-Tinhofer but not $(k + 1)$-Tinhofer.
    \item For every fixed integer $k \ge 0$, the problem of deciding whether a given $k$-Tinhofer graph is ($k + 1$)-Tinhofer is $\P$-hard under uniform $\ACz$ many-one reductions.
    \item We show that testing isomorphism between an $(n - k)$-Tinhofer graph $G$ and an arbitrary graph $H$ is fixed-parameter tractable with respect to the parameter $k$.
\end{itemize}
\end{abstract}

\newpage
\tableofcontents

\section{Introduction}
\label{sec:intro}

Graph isomorphism is the problem of determining whether two given graphs $G$ and $H$ are isomorphic, that is, whether there exists a bijection $\sigma : V(G) \to V(H)$ such that for all $(u,v) \in V(G) \times V(G)$, $(u,v) \in E(G)$ if and only if $(\sigma(u),\sigma(v)) \in E(H)$. This problem has been widely explored over the past several decades and is significant from both a theoretical and practical perspective. It is known to be in \NP, but is not known to be $\NP$-complete. The complexity status of graph isomorphism is intriguing because it is known that if the problem is $\NP$-complete, then the polynomial hierarchy collapses to the second level. In a major breakthrough, Babai~\cite{Bab16} gave a quasi-polynomial time algorithm for graph isomorphism running in time $2^{(\log {n})^{O(1)}}$, which was later refined to $2^{O((\log n)^3)}$ by Helfgott {\em et al.}~\cite{HHBD17}. Moreover, polynomial-time algorithms are known for several restricted graph classes, including planar graphs~\cite{HW74}, graphs of bounded degree~\cite{Luk82}, graphs of bounded treewidth~\cite{Bod90}, and graphs with bounded eigenvalue multiplicities~\cite{Bab82}. A well-known variant of graph isomorphism is the Colored Graph Isomorphism (CGI) problem, which asks whether two vertex-colored graphs $G$ and $H$ admit a bijection that preserves both adjacency and vertex colors. Despite the additional color constraints, CGI is computationally equivalent to the uncolored version.

The color refinement approach, also known as naive vertex classification or the $1$-dimensional Weisfeiler-Leman algorithm, is a combinatorial technique for graph isomorphism testing. Weisfeiler and Leman~\cite{WL68} originally introduced the $2$-dimensional, or \emph{classical} version, which refines colors on pairs of vertices rather than on individual vertices. The algorithm starts with an initial coloring of the vertices and iteratively refines the vertex coloring through a series of refinement steps, where each step preserves isomorphism, until no further refinement is possible. The resulting partition is called the \emph{stable partition}, and each part of the partition is called a \emph{cell}. Given two graphs $G$ and $H$, the algorithm compares the multisets of colors obtained after stabilization to determine whether the graphs are distinguishable by color refinement (see \cref{sec:prelims} for details). 

However, color refinement does not solve the isomorphism testing problem correctly on all graphs. In particular, there exist non-isomorphic graphs that color refinement fails to distinguish. This motivates the notion of amenable graphs. A graph $G$ is called \emph{amenable} if, for every graph $H \ncong G$, the color refinement algorithm distinguishes $G$ from $H$, that is, the stable partitions of $G$ and $H$ are different~\cite{AKRV17}. Arvind {\em et al.}~\cite{AKRV17} obtained a characterization of amenable graphs using structural constraints on the graph, which also led to a linear-time algorithm to check whether a graph is amenable or not. Graphs for which color refinement produces a trivial partition (singleton color classes) are called \emph{discrete graphs}, and every discrete graph is amenable~\cite{AKRV17}. Independently, Kiefer {\em et al.}~\cite{KSS21} gave a complete characterization of the graphs identified by color refinement. They also obtained a characterization of amenable graphs similar to that of~\cite{AKRV17}, and generalized the result to arbitrary relational structures, including directed graphs. In recent years, much work has been done to understand the power of color refinement for graph isomorphism testing~\cite{Ki20,K20,DH17,GKMS21}. Apart from graph isomorphism, color refinement has applications in several areas, including dimension reduction~\cite{GKMS14,KMGG14}, logic~\cite{KV15,IL90}, combinatorics~\cite{Dvo10,DGR18}, proof complexity~\cite{BG15}, graph neural networks~\cite{MRFL19,MLM23}, and graph kernels~\cite{KMGG14}.

Even though color refinement is not a complete test for graph isomorphism, it is used as a subroutine in almost all practical graph isomorphism tools~\cite{MB81}. To extend the power of color refinement further, the idea of \emph{individualization} was introduced~\cite{Tin91}. In this process, a vertex from a non-singleton color class is assigned a new unused color in both the graphs $G$ and $H$, and the color refinement procedure is applied again. Intuitively, this forces any isomorphism to map the chosen vertices to one another. The algorithm repeatedly applies individualization and refinement until either the graphs are distinguished or the partition becomes trivial (singleton color classes). If the multisets of colors in the resulting partitions are equal, then the graphs are declared isomorphic; otherwise, they are declared non-isomorphic. However, the algorithm may incorrectly conclude that two graphs are non-isomorphic even when they are actually isomorphic, because of a particular choice of vertices during individualization that destroys the existing isomorphism. On the other hand, if two graphs are non-isomorphic, then the algorithm always correctly distinguishes them. A graph $G$ is called a \emph{Tinhofer graph} if the algorithm correctly tests isomorphism between $G$ and every graph $H$ irrespective of the choices of vertices made during individualization~\cite{Tin91}. In particular, every discrete graph is Tinhofer by definition. 

 Arvind {\em et al.}~\cite{AKRV17} studied a hierarchy of various graph classes motivated by color refinement, namely Discrete, Amenable, Tinhofer, and Refinable graphs. A graph is called \emph{refinable} if the stable partition obtained by color refinement coincides with the orbit partition of the automorphism group of the graph. They showed the following containment relations among these graph classes.
\begin{equation}
    \textrm{Discrete} \subset \textrm{Amenable} \subset \textrm{Tinhofer} \subset \textrm{Refinable}
\end{equation}
 More importantly, they constructed graph examples to show that all the containments above are strict. Moreover, they also proved that the complexity of recognizing each of these graph classes is $\P$-hard~\cite{AKRV17}. \\ [-3mm]

\noindent {\bf Our Results:} Arvind {\em et al.}~\cite{AKRV17} showed that graph isomorphism can be tested in polynomial time for Tinhofer graphs. Motivated by this, we propose a new hierarchy of graph classes based on Tinhofer's algorithm, which we call the \emph{$k$-Tinhofer hierarchy}, with the goal of extending efficient isomorphism testing beyond Tinhofer graphs. For an integer $k \ge 0$, a graph $G$ is called \emph{$k$-Tinhofer} if, for every graph $H \cong G$, the colored graphs obtained after any sequence of $k$ rounds of individualization and refinement remain isomorphic, irrespective of the choices of vertices made during individualization (see the formal definition in \cref{sec:k-tin}). Intuitively, this hierarchy measures how many rounds of individualization and refinement can be performed while still preserving isomorphism. In particular, $0$-Tinhofer is the set of all graphs, while the class of Tinhofer graphs coincides with the class of $n$-Tinhofer graphs.

 Arvind {\em et al.}~\cite{AKRV17} proved that a graph $G$ is Tinhofer if and only if, for every $S \subseteq V(G)$, the partition of $V(G)$ into the orbit partition obtained by the action of the point-wise stabilizer subgroup of the automorphism group of $G$ stabilizing the set $S$ (denoted by $\operatorname{Orb}(\operatorname{Aut}_{S}(G))$) coincides with the stable partition obtained after individualizing the vertices of $S$ and applying color refinement (denoted by $P_{S}(G)$) (see \cref{sec:prelims} for the formal definitions of $\operatorname{Orb}(\operatorname{Aut}_{S}(G))$ and $P_{S}(G)$). Our first result is an extension of this characterization to the $k$-Tinhofer hierarchy. We show the following:

\begin{restatable}{theorem}{introthmcharacterization}
\label{introthmcharacterization}
Let $G$ be a graph with $|V(G)|=n$, and let $k$ be an integer with $n \ge k+1$. Then $G$ is ($k+1$)-Tinhofer if and only if, for every set $S \subseteq V(G)$ with $|S| \le k$, the partition of $V(G)$ into $\operatorname{Orb}(\operatorname{Aut}_{S}(G))$ coincides with the partition $P_{S}(G)$. 
\end{restatable}

 An immediate corollary of the above characterization is that the class of $1$-Tinhofer graphs coincides with the class of refinable graphs. Consequently, we obtain the following hierarchy of graph classes based on Tinhofer's algorithm and color refinement (see~\cref{fig:hierarchy}).

\begin{figure}[!ht]
    \centerline{\xymatrix{
    \text{Discrete} \ar[r] & \text{Amenable} \ar[r] & \text{Tinhofer} \ar@{=}[r] & \text{$n$-Tinhofer} \ar@{=}[d] \\
    \text{$0$-Tinhofer} \ar@{=}[d] & \text{$1$-Tinhofer} \ar[l]  \ar@{=}[d] & \text{ $\cdots$ } \ar[l] & \text{($n-1$)-Tinhofer} \ar@{->}[l]\\
    \text{All Graphs} & \text{Refinable} \ar[l] & & 
    }}
    \caption{Containment of Tinhofer Hierarchy}
    \label{fig:hierarchy}
\end{figure}
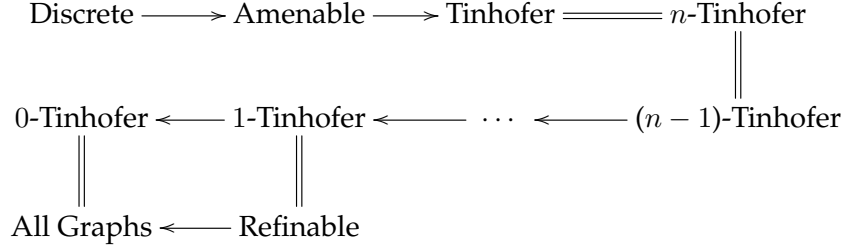

 The successive individualization and refinement steps in Tinhofer's algorithm can be represented by a tree, called the Individualization-Refinement tree (IR-tree) (studied in~\cite{ABS21}; see~\cref{sec:prelims} for a formal definition). A \emph{cell selector} is an isomorphism-invariant rule that, whenever the current coloring is not discrete, chooses a non-singleton color class (cell) for individualization. Each node of the IR-tree represents a sequence of individualized vertices, and corresponds to the stable coloring of the graph obtained after those individualizations. Such a stable coloring can also be represented by a \emph{quotient graph}, which records the adjacency relations between the color classes in the stable coloring (see~\cref{sec:prelims} for details). 

 Let $\pi$ and $\rho$ denote the stable colorings of graphs $G$ and $H$, respectively. Further, let $\pi_{\gamma}$ and $\rho_{\mu}$ denote the stable colorings obtained after individualizing the vertex sequences $\gamma = (v_1, \dots, v_k)$ and $\mu = (w_1, \dots, w_k)$ in $G$ and $H$, respectively. Let $Q(G, \pi_{\gamma})$ and $Q(H, \rho_{\mu})$ be the corresponding quotient graphs. For a fixed isomorphism-invariant cell selector $\mathrm{Sel}$, let $\Gamma_{\mathrm{Sel}}(G)$ denote the IR-tree of $G$. We prove the following combinatorial characterization of $k$-Tinhofer graphs in terms of IR-trees: 

\begin{restatable}{theorem}{introthmcharacterizationirtree}
\label{introthmcharacterizationirtree}
Let $G$ be a graph on $n$ vertices and $k$ be an integer with $n \ge k+1$. Then $G$ is $(k + 1)$-Tinhofer if and only if, for every graph $H \cong G$ and every isomorphism-invariant cell selector $\mathrm{Sel}$, the following holds: for every root-to-level-$k$ path $\gamma = (v_1, v_2, \dots, v_k)$ in $\Gamma_{\mathrm{Sel}}(G)$, there exists a corresponding root-to-level-$k$ path $\mu = (w_1, w_2, \dots, w_k)$ in $\Gamma_{\mathrm{Sel}}(H)$ such that $Q(G,\pi_{\gamma}) = Q(H,\rho_{\mu})$. 
\end{restatable}

 Intuitively, $G$ is $(k + 1)$-Tinhofer if and only if every sequence of $k$ individualizations in $G$ can be matched by a corresponding sequence in any graph $H \cong G$ so that the resulting quotient graphs remain identical. This characterization reformulates ($k+1$)-Tinhofer graphs in terms of IR-trees and quotient graphs, providing a combinatorial viewpoint that is convenient for analyzing the algorithm and for proving subsequent results.\\ [-3mm] 

\noindent {\bf Separation of the levels of the hierarchy:} 
We construct graph examples that separate successive graph classes in the above hierarchy (see~\cref{fig:hierarchy}), motivated by the framework of~\cite{AKRV17}.

\begin{restatable}{theorem}{introthmseparation}
    For every fixed integer $k \ge 0$, there exist vertex-colored graphs that are $k$-Tinhofer but are not $(k + 1)$-Tinhofer.
    \label{theo:introthmseparation}
\end{restatable}
The main technical challenge lies in extending the gadget construction so that it is sensitive to the levels of the hierarchy. \\[-3mm]
 
\noindent {\bf Membership testing for the hierarchy:} We now turn to the complexity of testing membership in these graph classes. It follows from the results of~\cite{AKRV17} that testing membership in any graph class $C$ with Discrete $\subseteq C \subseteq$ Refinable is $\P$-hard. Since for every fixed integer $k \ge 1$, the class of $k$-Tinhofer graphs lies between Discrete and Refinable (see~\cref{fig:hierarchy}), testing whether a graph is $k$-Tinhofer is $\P$-hard. We explore the complexity of each layer of this hierarchy. That is, for each fixed integer $k \ge 0$, we consider \ktintest: given a graph $G$ promised to be $k$-Tinhofer, decide whether $G$ is ($k+1$)-Tinhofer. Note that the problem is trivial for $k=n-1$ since the class of ($n-1$)-Tinhofer graphs coincides with the class of $n$-Tinhofer graphs.

\begin{restatable}{theorem}{introthmcomptintestphard}
For every fixed integer $k \ge 0$, \ktintest is $\P$-hard under uniform $\ACz$ many-one reductions.
\label{theo:introthmcomptintestphard}
\end{restatable}
 Here $k$ is fixed independently of the number of vertices in the graph. Note that when $k=0$, the above problem reduces to testing whether a graph is refinable, since the class of $1$-Tinhofer graphs coincides with the class of refinable graphs. Moreover, testing membership in refinable graphs is polynomial-time many-one reducible to graph isomorphism. Also, the problem is at least as hard as testing membership in vertex-transitive graphs, and hence at least as hard as isomorphism testing of vertex-transitive graphs (see~\cite[Section~8]{AKRV17} for details). On the hardness side, \cite{AKRV17} already shows $\P$-hardness when $k=0$. Indeed, we do not know a polynomial-time algorithm for the problem \ktintest. 

 The main technical challenge in proving the above hardness results lies in designing the gadgets that are sensitive to the different levels of the hierarchy.\\[-3mm]

\noindent {\bf FPT Algorithm:} We define the \emph{Tinhofer deficiency} of a graph $G$ on $n$ vertices to be the largest integer $\ell$ such that $G$ is not $(n-\ell)$-Tinhofer. Given a graph $G$ with Tinhofer deficiency $k$ and any graph $H$, we show that the graph isomorphism problem admits a fixed-parameter tractable (FPT) algorithm with respect to the parameter $k$. We record it as the following proposition: 

\begin{restatable}{proposition}{introfptalgoforktin}
\label{prop:introfptalgoforktin}
Given a graph $G$ on $n$ vertices with Tinhofer deficiency $k$, and any graph $H$, testing whether $G \cong H$ is fixed-parameter tractable with respect to the parameter $k$.  
\end{restatable}

 We present the algorithm in~\cref{sec:fpt}. In particular, this yields a polynomial-time algorithm when the Tinhofer deficiency is bounded by a constant.\footnote{Proposition~6.1 of~\cite{AFKKR22} handles the analogous problem $(n-k)$-Discrete and obtains a similar result. Both algorithms use brute-force on a bounded non-individualized portion of the graph.} This shows that the Tinhofer deficiency provides a natural structural parameter under which graph isomorphism becomes tractable. An interesting open problem is whether graph isomorphism admits a polynomial kernel parameterized by Tinhofer deficiency. 

\section{Preliminaries} 
\label{sec:prelims}
We use standard graph notation. For a graph $G$, we denote its vertex set and edge set by $V(G)$ and $E(G)$, respectively. The neighborhood of a vertex $u \in V(G)$ is denoted by $N_{G}(u)$. For technical convenience, we consider graphs to be vertex-colored throughout the paper. A \emph{vertex-colored} graph is a graph $G$ together with a coloring $c : V(G) \rightarrow \{1, 2, \dots, k\}$. Automorphisms of a vertex-colored graph and isomorphisms between vertex-colored graphs are required to preserve vertex colors. We get ordinary (uncolored) graphs when $c$ is constant. An \textit{automorphism} is an isomorphism from $G$ to itself. The set of all automorphisms of $G$ forms a group, denoted by $\operatorname{Aut}(G)$. The orbit of a vertex $u \in V(G)$ under $\operatorname{Aut}(G)$ is $\{v \in V(G) \mid \,  \exists \alpha \in \operatorname{Aut}(G), \alpha(u) = v\}$. This relation partitions $V(G)$ into equivalence classes called the \emph{orbit partition} of $G$. For $u \in V(G)$, let $\operatorname{Aut}_{u}(G)$ denote the stabilizer subgroup of $\operatorname{Aut}(G)$ that fixes $u$, and let $\operatorname{Orb}(\operatorname{Aut}_{u}(G))$ denote its orbit partition. More generally, for a set $S \subseteq V(G)$, we define $\operatorname{Aut}_S(G) = \bigcap_{u \in S} \operatorname{Aut}_u(G)$, and let $\operatorname{Orb}(\operatorname{Aut}_{S}(G))$ denote the corresponding orbit partition.\\ [-3mm]

\noindent {\bf Weisfeiler-Leman Color Refinement:}
Color refinement is a combinatorial algorithm that iteratively refines partitions of the vertices of a graph. For a graph $G$, the initial coloring $C^{0}$ is the vertex coloring of $G$, i.e., $C^{0}(v) = c(v)$. In the $(i+1)$-st iteration, the color assigned to a vertex $v$ is defined as: $C^{i+1}(v) = (C^{i}(v), \ldblbrace C^{i}(a) : a \in N_{G}(v) \rdblbrace)$, where $\ldblbrace \dots \rdblbrace$ denotes the multiset of colors of the neighbors of $v$ under $C^{i}$. The partition $P^{i+1}$ of $V(G)$ into the color classes of $C^{i+1}$ is a refinement of the partition $P^{i}$ induced by the coloring $C^{i}$. The algorithm stops when $P^{i + 1} = P^{i}$ for some iteration $i$. The resulting partition is called a \emph{stable partition}. An orbit partition is always a stable partition, but the converse is not true~\cite{KS09}. For isomorphism testing, given two graphs $G$ and $H$, color refinement is run simultaneously on both graphs. After obtaining the stable partitions, if the multisets of colors in $G$ and $H$ are different, then the graphs are not isomorphic. Otherwise, the algorithm declares them indistinguishable by color refinement. Arvind {\em et al.}~\cite{AKRV17} defined a graph $G$ to be \emph{amenable} if, for every graph $H \ncong G$, color refinement distinguishes $G$ from $H$.
\\ [-3mm]

\noindent {\bf Tinhofer's algorithm:} Tinhofer's algorithm tests isomorphism between two graphs $G$ and $H$ by repeatedly applying color refinement together with individualization. Individualizing a vertex $v$ means assigning it a new color distinct from all other vertices of the graph~\cite{Tin91}. More generally, individualizing a set of vertices means assigning distinct new colors to each vertex in the set, thereby distinguishing them from one another and from all other vertices of the graph. We denote by $P_{S}(G)$ the stable partition obtained after individualizing a set $S \subseteq V(G)$ and then applying color refinement on $G$. If the multisets of colors in $G$ and $H$ are different, then the algorithm concludes that $G \ncong H$. If the multisets of colors are equal and all color classes are singletons, then the algorithm concludes that $G \cong H$. Otherwise, the algorithm chooses a non-singleton color class (called a \emph{cell}) in both graphs, selects arbitrary vertices $u \in V(G)$ and $v \in V(H)$ from the corresponding cells, individualizes them with the same new color, and runs color refinement again. The process continues until the graphs are either distinguished or both stable partitions become discrete~\cite{AKRV17}. A graph $G$ is called a \emph{Tinhofer graph} if, for every graph $H \cong G$, Tinhofer's algorithm correctly concludes that $G \cong H$, irrespective of the choices of vertices made during individualization~\cite{Tin91}. \\ [-3mm]

\noindent \textbf{Individualization-Refinement Tree (IR-tree):} An IR-tree, denoted by $\Gamma_{\mathrm{Sel}}(G)$, represents all possible runs of individualization and refinement in Tinhofer's algorithm on a graph $G$, with respect to a chosen isomorphism-invariant cell selector $\mathrm{Sel}$~\cite{ABS21}. The nodes of $\Gamma_{\mathrm{Sel}}(G)$ are sequences $\gamma = (v_1, v_2, \dots, v_k)$ of vertices of $G$ individualized in order, starting from the empty sequence at the root. A branch $\gamma \rightarrow \gamma \cdot v$ corresponds to extending the sequence by individualizing a vertex $v$ from the cell chosen by the cell selector. Each node of the IR-tree represents a sequence of individualized vertices and corresponds to a stable coloring obtained after those individualizations. Leaves represent sequences that yield a discrete coloring (singleton color classes).  \\ [-3mm] 

 \noindent\textbf{Quotient graph:} For a stable coloring $\pi$ of a graph $G$, the quotient graph $Q(G,\pi)$ is a colored graph that encodes the adjacency relations between the color classes of $\pi$. The quotient graph is viewed as a directed graph. The vertex set of $Q(G, \pi)$ corresponds to the color classes of $\pi$. These vertices are colored with the color of the cell they represent in $G$. An edge from a color class $C_i$ to a color class $C_j$ is labeled with the number of neighbors a vertex of $C_i$ has in $C_j$. Since $\pi$ is a stable coloring, all vertices of $C_i$ have the same number of neighbors in $C_j$~\cite{ABS21}.\\[-3mm]  

\noindent{\bf CFI gadget:} A CFI gadget, denoted by $X_{k}$, is a vertex-colored graph with $k$ external vertex pairs \{$P_{1}, P_{2}, \dots, P_{k}$\} and $2^{k-1}$ intermediate vertices. Each pair $P_i = \{ a_i, b_i  \}$ is assigned a distinct color, and all intermediate vertices, denoted by $F$, share a common color distinct from those of the external pairs~\cite{CFI92}. We refer to $P_{1}$ and $P_{2}$ as the input vertex pairs, and $P_{3}, P_{4}, \dots, P_{k}$ as the output vertex pairs of the CFI gadget $X_{k}$. The set $F$ consists of all binary strings of length $k$ with an even number of ones from the set $\{0, 1\}^{k}$. Each vertex in $F$ is adjacent to $a_{i}$ in $P_{i}$ if the $i^{th}$ bit of the string is $0$, and to $b_i$ otherwise. Thus, every intermediate vertex is adjacent to an even number of the vertices $b_{i}$. The CFI gadgets $X_{3}$ and $X_{4}$ are illustrated in Fig. \ref{fig:cfi-x-3} and \ref{fig:cfi-x-4} respectively.

 An automorphism $\phi \in Aut(X_{k})$ is said to flip a vertex pair $P_{i} = \{ a_{i}, b_{i} \}$ if $\phi(a_{i}) = b_{i}$ and $\phi(b_{i}) = a_{i}$. Otherwise, $\phi$ is said to fix the pair $P_i$ if $\phi(a_{i}) = a_{i}$ and $\phi(b_{i}) = b_{i}$. Cai {\em et al.}~\cite{CFI92} proved that, when each pair $P_i = \{a_i, b_i\}$ is assigned a distinct color, every color-preserving automorphism of the CFI gadget $X_k$ flips an even number of the pairs \{$P_{1}, P_{2}, \dots, P_{k}$\}. This parity property of the CFI gadget plays a central role in the graph constructions used for our separation and membership results.  \\[-3mm]

\noindent {\bf IMP gadget:} An IMP gadget $Y_{k}$ is a modification of the CFI gadget consisting of an additional vertex pair $P_{0}$, together with the $k$ vertex pairs $\{P_{1}, P_{2}, \dots, P_{k}\}$ and the $2^{k-1}$ intermediate vertices of the CFI gadget. The first two pairs $P_{1}$ and $P_{2}$ are connected to the vertex pair $P_{0}$ by two edges~\cite{AFKKR22}. The IMP gadget $Y_{3}$ is illustrated in \cref{fig:imp-y-3}.

\begin{figure}
\centering
\subfigure[CFI($P_1$, $P_2$, $P_3$)]{\begin{tikzpicture}[scale=1.4, every node/.style={circle, draw, fill=black, inner sep=1.5pt}]
\node (Pi1) at (0,0) {}; \node (Pi2) at (0.5,0) {};
\node (Pj1) at (1.5,0) {}; \node (Pj2) at (2,0) {};
\node (Fk1) at (0.2,1) {}; \node (Fk2) at (0.7,1) {};
\node (Fk3) at (1.2,1) {}; \node (Fk4) at (1.7,1) {}; 
\node (Pk1) at (0.85,2) {};
\node (Pk2) at (1.25,2) {}; 
\begin{scope}[on background layer]
    \fill[gray!30, rounded corners] (-0.2,-0.2) rectangle (0.7,0.2); 
    \fill[gray!30, rounded corners] (1.3,-0.2) rectangle (2.2,0.2); 
    \fill[gray!30, rounded corners] (0,0.8) rectangle (1.9,1.2); 
    \fill[gray!30, rounded corners] (0.7,1.8) rectangle (1.4,2.2); 
\end{scope}
\node[draw=none, fill=none] at (0.3,-0.4) {$P_1$}; 
\node[draw=none, fill=none] at (1.8,-0.4) {$P_2$};
\node[draw=none, fill=none] at (1,2.4) {$P_3$}; 
\node[draw=none, fill=none] at (1.7,1.3) {$F$}; 
\draw (Pi1) -- (Fk1); \draw (Pi1) -- (Fk2);
\draw (Pi2) -- (Fk3); \draw (Pi2) -- (Fk4);
\draw (Pj1) -- (Fk1); \draw (Pj1) -- (Fk3);
\draw (Pj2) -- (Fk2); \draw (Pj2) -- (Fk4);
\draw (Pk1) -- (Fk1); \draw (Pk1) -- (Fk4);
\draw (Pk2) -- (Fk2); \draw (Pk2) -- (Fk3);
\end{tikzpicture}
\label{fig:cfi-x-3} }
\hspace{2mm}
\subfigure[CFI($P_1$, $P_2$, $P_3, P_4$)]{\begin{tikzpicture}[scale=1.4, every node/.style={circle, draw, fill=black, inner sep=1.5pt}]
\centering
\node (Pi1) at (0,0) {}; \node (Pi2) at (0.5,0) {};
\node (Pj1) at (3.4,0) {}; \node (Pj2) at (3.9,0) {};
\node (Fk1) at (0.2,1) {}; \node (Fk2) at (0.7,1) {};
\node (Fk3) at (1.2,1) {}; \node (Fk4) at (1.7,1) {}; 
\node (Fk5) at (2.2,1) {};  \node (Fk6) at (2.7,1) {}; 
\node (Fk7) at (3.2,1) {};  \node (Fk8) at (3.7,1) {}; 
\node (Pk1) at (0.6,2) {}; 
\node (Pk2) at (1.1,2) {};
\node (Pk3) at (2.8,2) {};  \node (Pk4) at (3.3,2) {};
\begin{scope}[on background layer]
    \fill[gray!30, rounded corners] (-0.2,-0.2) rectangle (0.7,0.2); 
    \fill[gray!30, rounded corners] (3.2,-0.2) rectangle (4.1,0.2);
    \fill[gray!30, rounded corners] (0,0.8) rectangle (3.9,1.2); 
    \fill[gray!30, rounded corners] (0.4,1.8) rectangle (1.3,2.2);
    \fill[gray!30, rounded corners] (2.6,1.8) rectangle (3.5,2.2); 
\end{scope}
\node[draw=none, fill=none] at (0.3,-0.4) {$P_1$}; 
\node[draw=none, fill=none] at (3.7,-0.4) {$P_2$}; 
\node[draw=none, fill=none] at (0.8,2.4) {$P_3$}; 
\node[draw=none, fill=none] at (3,2.4) {$P_4$};
\node[draw=none, fill=none] at (3.8,1.3) {$F$}; 
\draw (Pi1) -- (Fk1); \draw (Pi1) -- (Fk2);
\draw (Pi1) -- (Fk3); \draw (Pi1) -- (Fk4);
\draw (Pi2) -- (Fk5); \draw (Pi2) -- (Fk6);
\draw (Pi2) -- (Fk7); \draw (Pi2) -- (Fk8);
\draw (Pj1) -- (Fk1); \draw (Pj1) -- (Fk2);
\draw (Pj1) -- (Fk5); \draw (Pj1) -- (Fk6);
\draw (Pj2) -- (Fk3); \draw (Pj2) -- (Fk4);
\draw (Pj2) -- (Fk7); \draw (Pj2) -- (Fk8);
\draw (Pk1) -- (Fk1); \draw (Pk1) -- (Fk3);
\draw (Pk1) -- (Fk5); \draw (Pk1) -- (Fk7);
\draw (Pk2) -- (Fk2); \draw (Pk2) -- (Fk4);
\draw (Pk2) -- (Fk6); \draw (Pk2) -- (Fk8);
\draw (Pk3) -- (Fk1); \draw (Pk3) -- (Fk4);
\draw (Pk3) -- (Fk6); \draw (Pk3) -- (Fk7);
\draw (Pk4) -- (Fk2); \draw (Pk4) -- (Fk3);
\draw (Pk4) -- (Fk5); \draw (Pk4) -- (Fk8);
\end{tikzpicture}
\label{fig:cfi-x-4} }
\hspace{2mm}
\subfigure[IMP gadget $Y_3$]{
\begin{tikzpicture}[scale=1.1, every node/.style={circle, draw, fill=black, inner sep=1.5pt}]
\node (Pi1) at (0,0) {};  \node (Pi2) at (0.5,0) {};
\node (Pj1) at (2,0) {};   \node (Pj2) at (2.5,0) {};
\node (Fk1) at (0.3,1) {};  \node (Fk2) at (1,1) {};
\node (Fk3) at (1.7,1) {};  \node (Fk4) at (2.4,1) {}; 
\node (Pk1) at (1,2) {};
\node (Pk2) at (1.7,2) {}; 
\node (P41) at (1, -1) {};  \node (P42) at (1.7, -1) {};
\begin{scope}[on background layer]
    \fill[gray!30, rounded corners] (-0.2,-0.2) rectangle (0.7,0.2); 
    \fill[gray!30, rounded corners] (1.8,-0.2) rectangle (2.7,0.2); 
    \fill[gray!30, rounded corners] (0.1,0.8) rectangle (2.6,1.2); 
    \fill[gray!30, rounded corners] (0.8,1.8) rectangle (2.1,2.2); 
    \fill[gray!30, rounded corners] (0.8,-0.8) rectangle (1.9,-1.2); 
\end{scope}
\node[draw=none, fill=none] at (0.1,-0.4) {$P_1$}; 
\node[draw=none, fill=none] at (2.4,-0.4) {$P_2$}; 
\node[draw=none, fill=none] at (1.4,2.4) {$P_3$};  
\node[draw=none, fill=none] at (2.3,1.3) {$F$}; 
\node[draw=none, fill=none] at (1.5,-1.4) {$P_{0}$}; 
\draw (Pi1) -- (Fk1); \draw (Pi1) -- (Fk2);
\draw (Pi2) -- (Fk3);  \draw (Pi2) -- (Fk4);
\draw (Pj1) -- (Fk1);  \draw (Pj1) -- (Fk3);
\draw (Pj2) -- (Fk2);  \draw (Pj2) -- (Fk4);
\draw (Pk1) -- (Fk1);  \draw (Pk1) -- (Fk4);
\draw (Pk2) -- (Fk2);  \draw (Pk2) -- (Fk3);
\draw (Pi1) -- (P41);  \draw (Pi2) -- (P42);
\draw (Pj1) -- (P41);  \draw (Pj2) -- (P42);
\end{tikzpicture}
\label{fig:imp-y-3} }
\caption{Illustration of the CFI gadgets $X_{3}, X_{4}$ and IMP gadget $Y_3$}
\label{fig:cfi-3-4-imp}
\end{figure}
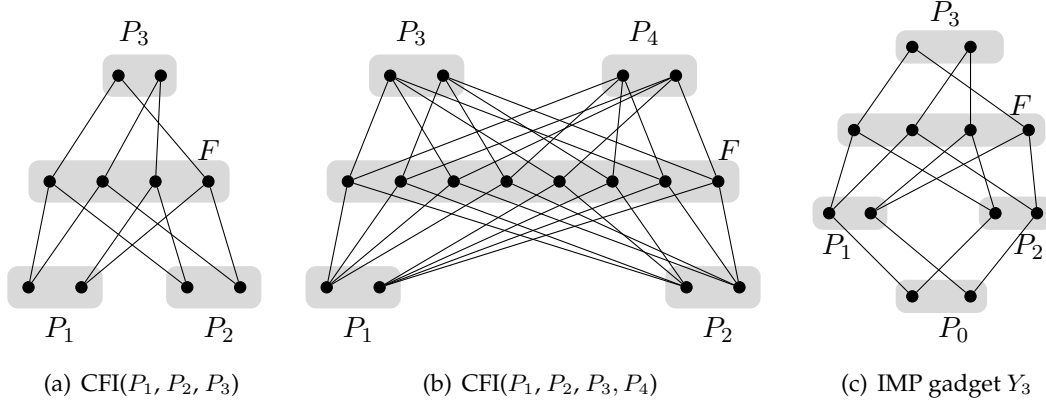

\section{\texorpdfstring{$k$}{k}-Tinhofer graphs - Basic Properties and Characterization}
\label{sec:k-tin}
Let $G$ and $H$ be graphs. We run Tinhofer's algorithm on $G$ and $H$, starting with an initial round of color refinement. Subsequently, each step consists of selecting vertices in the two graphs from the same color class, individualizing them, and applying color refinement. Let $G^{(k)}$ and $H^{(k)}$ denote the colored graphs obtained after $k$ such steps. For $k \in \{0, \dots n\}$, a graph $G$ is called \emph{k-Tinhofer} if for every graph $H$ with $H \cong G$, we have $G^{(k)} \cong H^{(k)}$ for every choice of vertices individualized at each of the $k$ steps. 
Based on this definition, we obtain the following properties of $k$-Tinhofer graphs:

\begin{lemma} The class of $k$-Tinhofer graphs satisfies the following properties:
\label{lemma:obs-k-tin}
\begin{itemize}  \let\labelitemi\labelitemii
\itemsep -2pt
\item If a graph $G$ is $n$-Tinhofer, then $G$ is $k$-Tinhofer, where $1 \le k \le (n-1)$. 
\item If a graph $G$ is $n$-Tinhofer, then $G$ is Tinhofer. Thus, Tinhofer = $n$-Tinhofer.
\item $(k + 1)$-Tinhofer $\subset k$-Tinhofer, where $1 \le k \le (n-1)$.
\item ($n-1$)-Tinhofer = n-Tinhofer.
\item $0$-Tinhofer contains all graphs.
\end{itemize}
\end{lemma}
\begin{proof}
\begin{itemize} \let\labelitemi\labelitemii
\item Let $G$ be a graph. Suppose $G$ is $n$-Tinhofer but not $k$-Tinhofer for some $k$, where $1 \le k \le (n-1)$. Then there exists a graph $H$ such that $G \cong H$, but $G^{(k)} \ncong H^{(k)}$. Hence, after $k$ steps, the algorithm fails to recognize the isomorphism. This failure still persists in the subsequent steps, contradicting that $G$ is $n$-Tinhofer.
\item If a graph $G$ is $n$-Tinhofer, then after $n$ steps all the vertices are individualized, so the algorithm correctly identifies isomorphism. Hence, $G$ is Tinhofer.
\item Suppose $G$ is $(k + 1)$-Tinhofer but not $k$-Tinhofer for some $1 \le k \le (n-1)$. Then, there exists a graph $H$ such that $G \cong H$, but $G^{(k)} \ncong H^{(k)}$. Since further steps of the algorithm preserve non-isomorphism, we have $G^{(k+1)} \ncong H^{(k+1)}$, which contradicts the assumption that $G$ is $(k + 1)$-Tinhofer.  
\item The inclusion $n$-Tinhofer $\subset$ $(n-1)$-Tinhofer follows from the previous point. Conversely, if $G$ is ($n-1)$-Tinhofer, then only one vertex remains to be individualized, and hence $G$ is $n$-Tinhofer.  
\item $0$-Tinhofer contains all graphs, since color refinement without individualization preserves isomorphism.
\end{itemize}
\end{proof}\vspace{-3mm}

\subsection{Algebraic characterization of \texorpdfstring{$k$}{k}-Tinhofer graphs}
Arvind {\em et al.} \cite{AKRV17} proved an algebraic characterization of Tinhofer graphs based on the pointwise stabilizer subgroup of $\operatorname{Aut}(G)$. We extend this characterization to the $k$-Tinhofer hierarchy. In particular, one direction of the proof follows a similar inductive argument on the number of individualized vertices as in Lemma~7.8 of~\cite{AKRV17}. We prove the following:

\introthmcharacterization*

\begin{proof}
($\rightarrow$) Let a graph $G$ be ($k + 1$)-Tinhofer. The orbit partition of $\operatorname{Aut}_{S}(G)$ is always a refinement of $P_{S}(G)$. We use a proof by contradiction to show the other inclusion. Suppose there exist two vertices $u, v$ which belong to the same cell of $P_{S}(G)$ but $\nexists \sigma \in \operatorname{Aut}_{S}(G)$ such that $\sigma(u) = v$. We create two isomorphic copies $G'$ and $G''$ of the graph $G$. We run Tinhofer's algorithm on both $G'$ and $G''$, and individualize all vertices in $S$ one-by-one, as per the isomorphism. This process must result in the same stable partition $P_{S}(G)$ in both graphs. Next, we individualize another vertex $u \in V(G')$ and $v \in V(G'')$. Now, the vertices $u$ and $v$ are assigned the same unique color, but there is no isomorphism between $G'$ and $G''$ that maps $u$ to $v$. Hence, $G'^{(k+1)} \ncong G''^{(k+1)}$. Thus, $G$ is not ($k+1$)-Tinhofer.\vspace{1mm} \\        
 ($\leftarrow$) For the other direction, suppose that for all $S \subseteq V(G)$ with $|S| \le k$, the two partitions are equal. Let $H$ be a graph such that $H \cong G$. We run Tinhofer's algorithm on both the graphs. Let $G_{0} = G$ and $H_{0} = H$. For $0 \le i \le k+1$, let $G_{i}$ and $H_{i}$ denote the colored graphs obtained after $i$ steps of individualization. We prove by induction on $i$ that $G_{i} \cong H_{i}$. The base case $i=0$ is trivial. Assume that $G_{i} \cong H_{i}$ via a mapping $\theta$, so that $\theta(u_{j}) = v_{j}$, for all $1 \le j \le i$. Let $u_{i+1} \in V(G)$ and $v_{i+1} \in V(H)$ be the vertices individualized in the $(i+1)^{\text{th}}$ step such that $\theta(u_{i+1}) \ne v_{i+1}$. Since $u_{i+1}$ and $v_{i+1}$ are assigned the same unique color, it follows that $u_{i+1}$ and $\theta^{-1}(v_{i+1})$ lie in the same color class of $P_{S}(G)$, where $S = \{u_1, \dots u_i\}$. Thus, $\exists \beta \in \operatorname{Aut}_{S}(G)$ such that $\beta(u_{i+1}) = \theta^{-1}(v_{i+1})$. Hence, $\theta \circ \beta$ maps $u_{i+1}$ to $v_{i+1}$, and therefore $G_{i+1} \cong H_{i+1}$. Thus, $G$ is $(k + 1)$-Tinhofer.
\end{proof}

 An immediate corollary is that the class of $1$-Tinhofer graphs coincides with the class of refinable graphs. 
 \begin{corollary}
 \label{coroll:1-tin-ref}
     $1$-Tinhofer = Refinable
 \end{corollary}

\subsection{Combinatorial characterization of \texorpdfstring{$k$}{k}-Tinhofer graphs}
Anders {\em et al.}~\cite{ABS21} defined Individualization-Refinement tree (IR-tree) in the context of Tinhofer graphs. Using this framework, we first state a structural lemma relating colored graphs and their quotient graphs, and then use it to prove the following characterization of $k$-Tinhofer graphs.
\begin{lemma}
Let $G$ and $H$ be graphs, and let $(G, \pi_{\gamma})$ and $(H, \rho_{\mu})$ be the colored graphs obtained after performing sequences $\gamma = (v_1, v_2, \dots, v_k)$ and $\mu = (w_1, w_2, \dots, w_k)$ of $k$ individualizations and refinements in $G$ and $H$. Then $(G,\pi_{\gamma}) \cong (H, \rho_{\mu})$ if and only if $Q(G, \pi_{\gamma}) = Q(H, \rho_{\mu})$.
    \label{lemma:quot-graph}
\end{lemma}

\begin{proof}
Suppose $(G,\pi_{\gamma}) \cong (H, \rho_{\mu})$. Then there exists a color-preserving isomorphism $\phi : (G, \pi_{\gamma}) \rightarrow (H, \rho_{\mu})$ that maps each color class of $\pi_{\gamma}$ to the corresponding color class of $\rho_{\mu}$, preserving color class sizes and neighbor counts. Hence the corresponding quotient graphs coincide. On the other hand, suppose $Q(G, \pi_{\gamma}) = Q(H, \rho_{\mu})$. Then there is a bijection between the color classes preserving color class sizes and adjacency counts. By the properties of stable colorings, this yields a color-preserving isomorphism between $(G,\pi_{\gamma})$ and $(H, \rho_{\mu})$ by mapping vertices within corresponding color classes according to the adjacency structure encoded by the quotient graph. Hence, $(G,\pi_{\gamma}) \cong (H, \rho_{\mu})$.    
\end{proof}

\introthmcharacterizationirtree*

\begin{proof}
Suppose $G$ is $(k + 1)$-Tinhofer and let $H \cong G$. Consider any root-to-level-$k$ path $\gamma = (v_1, v_2, \dots, v_k)$ in $\Gamma_{\mathrm{Sel}}(G)$. Since $G$ is ($k+1$)-Tinhofer, there exists a sequence $\mu = (w_1, w_2, \dots, w_k)$ in $H$ such that $(G, \pi_{\gamma}) \cong (H, \rho_{\mu})$. By~\cref{lemma:quot-graph}, this implies $Q(G, \pi_{\gamma}) = Q(H, \rho_{\mu})$. 

 Conversely, suppose the stated condition holds. Let $H \cong G$. Consider any sequence of $k$ individualizations performed by Tinhofer's algorithm on $G$, yielding a path $\gamma = (v_1, v_2, \dots, v_k)$ in $\Gamma_{\mathrm{Sel}}(G)$. By assumption, there exists a corresponding path $\mu = (w_1, \dots, w_k)$ in $\Gamma_{\mathrm{Sel}}(H)$ such that $Q(G, \pi_{\gamma}) = Q(H, \rho_{\mu})$. By~\cref{lemma:quot-graph}, $(G, \pi_{\gamma}) \cong (H, \rho_{\mu})$. Hence, Tinhofer's algorithm correctly identifies the graphs as isomorphic after $k$ steps of individualization and refinement, implying that $G$ is ($k+1$)-Tinhofer.  
\end{proof}

\section{\texorpdfstring{$k$}{k}-Tinhofer Hierarchy is Strict}
\label{sec:k-tin-strict}
In this section, we construct examples of graphs that separate each of the successive graph classes of the hierarchy, i.e. $k$-Tinhofer from $(k + 1)$-Tinhofer graphs (see \cref{fig:hierarchy}). We assume all graphs are vertex-colored as defined in~\cref{sec:prelims}, and all isomorphisms are color-preserving.   

 The case for $k=0$ follows since $0$-Tinhofer is the class of all graphs, while $1$-Tinhofer is exactly the class of refinable graphs, and there exist graphs (e.g., the $12$-vertex Frucht graph~\cite{F49}) that are not refinable. For every fixed integer $k \ge 1$, we construct a graph $H$ that will witness the separation.

\noindent \textbf{Construction of graph $H$:} We construct the graph $H$ by combining the CFI gadgets $X_{3}$ and $X_{k+2}$ (see \cref{fig:cfi-3-4-imp} in~\cref{sec:prelims}) to show the separation. The input pairs of both the gadgets are shared, while their output pairs remain different. Let $P_{1}$ and $P_{2}$ be the shared input pairs. The gadget $X_3$ contributes the output pair $P_{0}$, and $X_{k+2}$ contributes the output pairs $P_{3}, P_{4}, \dots, P_{k+2}$, giving $k$ output pairs in total (see~\cref{fig:seperating-example} for an illustration of the construction). By the properties of the gadgets~\cite{CFI92}, the pairs $P_{0}, P_{1}, \dots, P_{k+2}$ and the intermediate vertex sets $F$, $F'$ form distinct color classes. We first establish a structural property of the constructed graph, and then use it to construct our separating example.

\begin{figure}[!ht]
\centering
\begin{tikzpicture}[scale=1.8, every node/.style={circle, draw, fill=black, inner sep=2pt}]
\node (Pi1) at (0.4,0) {}; \node (Pi2) at (0.9,0) {};
\node (Pj1) at (3.6,0) {}; \node (Pj2) at (4.2,0) {}; 
\node (Fk11) at (-0.2,1) {}; \node (Fk12) at (0.3,1) {};
\node (Fk13) at (0.8,1) {}; \node (Fk14) at (1.3,1) {}; 
\node (Fk1) at (2.3,1) {}; \node (Fk2) at (2.8,1) {};
\node (Fk3) at (3.3,1) {}; \node (Fk4) at (3.8,1) {}; 
\node (Fk5) at (4.3,1) {};  \node (Fk6) at (4.8,1) {}; 
\node (Fk7) at (5.3,1) {};  \node (Fk8) at (5.8,1) {}; 
\node (Pk11) at (0.6,2) {}; 
\node (Pk12) at (1.1,2) {}; 
\node (Pk1) at (2.6,2) {};
\node (Pk2) at (3.1,2) {};
\node (Pk3) at (4.6,2) {};
\node (Pk4) at (5.1,2) {};
\begin{scope}[on background layer]
    \fill[gray!30, rounded corners] (0.2,-0.2) rectangle (1.1,0.2); 
    \fill[gray!30, rounded corners] (3.4,-0.2) rectangle (4.4,0.2);
    \fill[gray!30, rounded corners] (-0.4,0.8) rectangle (1.5,1.2); 
    \fill[gray!30, rounded corners] (2.1,0.8) rectangle (6,1.2); 
    \fill[gray!30, rounded corners] (0.4,1.8) rectangle (1.3,2.2);
    \fill[gray!30, rounded corners] (2.4,1.8) rectangle (3.3,2.2); 
    \fill[gray!30, rounded corners] (4.4,1.8) rectangle (5.3,2.2); 
\end{scope}
\node[draw=none, fill=none] at (0.7,-0.4) {$P_1$};
\node[draw=none, fill=none] at (3.9,-0.4) {$P_2$}; 
\node[draw=none, fill=none] at (0.9,2.4) {$P_0$}; 
\node[draw=none, fill=none] at (2.9,2.4) {$P_3$}; 
\node[draw=none, fill=none] at (4.9,2.4) {$P_4$}; 
\node[draw=none, fill=none] at (1.4,1.3) {$F$}; 
\node[draw=none, fill=none] at (5.9,1.3) {$F'$}; 
\draw (Pi1) -- (Fk11); \draw (Pi1) -- (Fk12);
\draw (Pi2) -- (Fk13); \draw (Pi2) -- (Fk14);
\draw (Pj1) -- (Fk11); \draw (Pj1) -- (Fk13);
\draw (Pj2) -- (Fk12); \draw (Pj2) -- (Fk14);
\draw (Pk11) -- (Fk11); \draw (Pk11) -- (Fk14);
\draw (Pk12) -- (Fk12); \draw (Pk12) -- (Fk13);
\draw (Pi1) -- (Fk1); \draw (Pi1) -- (Fk2);
\draw (Pi1) -- (Fk3); \draw (Pi1) -- (Fk4);
\draw (Pi2) -- (Fk5); \draw (Pi2) -- (Fk6);
\draw (Pi2) -- (Fk7); \draw (Pi2) -- (Fk8);
\draw (Pj1) -- (Fk1); \draw (Pj1) -- (Fk2);
\draw (Pj1) -- (Fk5); \draw (Pj1) -- (Fk6);
\draw (Pj2) -- (Fk3); \draw (Pj2) -- (Fk4);
\draw (Pj2) -- (Fk7); \draw (Pj2) -- (Fk8);
\draw (Pk1) -- (Fk1); \draw (Pk1) -- (Fk3);
\draw (Pk1) -- (Fk5); \draw (Pk1) -- (Fk7);
\draw (Pk2) -- (Fk2); \draw (Pk2) -- (Fk4);
\draw (Pk2) -- (Fk6); \draw (Pk2) -- (Fk8);
\draw (Pk3) -- (Fk1); \draw (Pk3) -- (Fk4);
\draw (Pk3) -- (Fk6); \draw (Pk3) -- (Fk7);
\draw (Pk4) -- (Fk2); \draw (Pk4) -- (Fk3);
\draw (Pk4) -- (Fk5); \draw (Pk4) -- (Fk8);
\end{tikzpicture}
\caption{An illustration of graph $H$ separating $2$-Tinhofer from $3$-Tinhofer. The left side and right side of the graph $H$ contain a combination of the CFI graphs $X_{3}$ and $X_{4}$ in such a way that they share the input pairs \{$P_{1}, P_{2}$\} (for $k = 2$).}
\label{fig:seperating-example}
\end{figure}
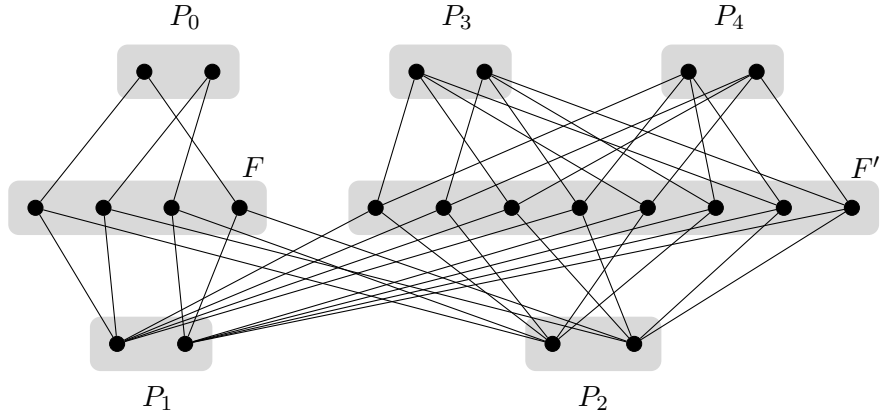    

\begin{lemma}
The vertex pair $P_{0}$ is flipped if and only if the vertex pairs $P_{3}, P_{4}, \dots, P_{k+2}$ are flipped an odd number of times in $H$. 
\label{lemma:k-flip-parity}
\end{lemma}
\begin{proof}
Since $X_{3}$ is a CFI gadget, flipping $P_{0}$ forces exactly one of $P_{1}$ and $P_{2}$ to be flipped. Again, since $X_{k+2}$ is a CFI gadget, if exactly one of the vertex pairs $P_1$ and $P_2$ is flipped, then the vertex pairs $P_{3}, P_{4}, \dots, P_{k+2}$ must be flipped an odd number of times.    
\end{proof}

\begin{lemma}
\label{lemma:strict-k-tinhofer}
For every fixed integer $k \ge 1$, the graph $H$ is $k$-Tinhofer but not $(k + 1)$-Tinhofer.
\end{lemma}
\begin{proof}
We consider two identical copies $H'$ and $H''$ of $H$ and run Tinhofer's algorithm on them. We claim that $H$ is $k$-Tinhofer. After any sequence of $k$ individualizations performed on corresponding color classes of $H'$ and $H''$, the CFI gadgets $X_{3}$ and $X_{k+2}$ still satisfy the even-flip property. In particular, the induced colorings of $H'$ and $H''$ admit a color-preserving isomorphism, since the parity constraints on flips are preserved in both copies. Hence, the resulting colored graphs remain isomorphic after $k$ steps.

 To show that $H$ is not $(k + 1)$-Tinhofer, we consider a sequence of individualizations for which the resulting colored graphs are not isomorphic after the $(k+1)$-st step. First, we individualize the same vertex in the color class $P_{0}$ in both graphs $H'$ and $H''$. The subsequent refinement step splits only $F$, and the resulting partition is stable. Next, we individualize the same vertex in the color class $P_{3}$ in both graphs $H'$ and $H''$. Again, the refinement step splits only $F'$. We continue in the same way for the next $(k-2)$ rounds, i.e., we individualize the same vertex in color classes $P_{4}, P_{5}, \dots, P_{k+1}$. The subsequent refinement further splits $F'$, whereas $P_{1}, P_{2}$ and $P_{k+2}$ do not split. Now, in the ($k+1$)-st round, we individualize color class $P_{k+2}$, but choose different vertices of $P_{k+2}$ in $H'$ and $H''$. The refinement step again splits only $F'$, whereas $P_{1}$ and $P_{2}$ are not yet split. At this point, the two colored graphs have $P_{0}, P_{3}, P_{4}, \dots, P_{k+1}$ unflipped, while $P_{k+2}$ is flipped. Hence, by~\cref{lemma:k-flip-parity}, these two colored graphs are not isomorphic. But $H'$ and $H''$ are identical copies of $H$, so this is a contradiction. Thus, $H$ is not $(k + 1)$-Tinhofer. 
\end{proof}      

\section{Complexity of \ktintest}
\label{sec:complex-k-tin}
In this section, we explore the complexity of \ktintest (see \cref{sec:intro} for definition). In particular, we show that \ktintest is \P-hard under uniform $\ACz$ many-one reductions from the Monotone Circuit Value Problem (\MCVP), which is known to be $P$-complete~\cite{Go77}. An instance of \MCVP\ consists of a monotone Boolean circuit $C$ with constant input gates and AND/OR gates, and asks whether $C$ evaluates to $1$. We assume all graphs are vertex-colored as defined in~\cref{sec:prelims}, and all isomorphisms are color-preserving.

 Our proof is motivated by the reduction in~\cite{AKRV17}, where a graph $G$ is constructed such that if $C$ evaluates to $1$, then $G$ is discrete, and if $C$ evaluates to $0$, then $G$ is not even refinable. This establishes the hardness result for the case $k=0$. For $k \ge 1$, we extend this construction to obtain a graph $N$ such that if $C$ evaluates to $1$, then $N$ is discrete, whereas if $C$ evaluates to $0$, then $N$ is $k$-Tinhofer but not $(k + 1)$-Tinhofer.

\medskip

\noindent \textbf{Construction of graph $N$:} Given a monotone Boolean circuit $C$, we construct a graph $N$ as follows: 
 Each gate $g_k$ in $C$ is represented by a vertex pair $P_{k} = \{a_{k}, b_{k}\}$. If $g_{k}$ is a constant $0$, then $a_{k}$ and $b_{k}$ receive the same color, whereas, if $g_{k}$ is a constant $1$, they are colored differently. If $g_{k}$ is an AND gate with input gates $g_{i}$ and $g_{j}$, it is replaced by a CFI gadget $X_3 = (P_{i}, P_{j}, P_{k})$ in $G$ (see~\cref{fig:cfi-x-3}). This construction ensures that the pair $P_k$ is refined if and only if both the input pairs $P_i$ and $P_j$ are refined. If $g_{k}$ is an OR gate with input gates $g_{i}$ and $g_{j}$, then it is replaced by two IMP gadgets that share the output pair $P_k$ (see~\cref{fig:or-gate} for an illustration of the gadget). This construction ensures that the pair $P_k$ is refined if and only if at least one of the input pairs $P_i$ or $P_j$ is refined. By induction on the height of the circuit, color refinement refines $P_{i} = \{a_{i}, b_{i}\}$ if and only if the corresponding gate $g_{i}$ evaluates to $1$.

    \begin{figure}[!ht]
\centering
\begin{tikzpicture}[scale=1.8, every node/.style={circle, draw, fill=black, inner sep=2pt}]
\node (Pi1) at (0,0) {};
\node (Pi2) at (0.5,0) {};
\node (Pj1) at (1.5,0) {};
\node (Pj2) at (2,0) {};
\node (Fk1) at (0.2,1) {};
\node (Fk2) at (0.7,1) {};
\node (Fk3) at (1.2,1) {};
\node (Fk4) at (1.7,1) {}; 
\node (P41) at (0.7, -1) {};
\node (P42) at (1.2, -1) {};
\node (Pi11) at (2.5,0) {};
\node (Pi12) at (3,0) {};
\node (Pj11) at (4,0) {};
\node (Pj12) at (4.5,0) {};
\node (Fk11) at (2.7,1) {};
\node (Fk12) at (3.2,1) {};
\node (Fk13) at (3.7,1) {};
\node (Fk14) at (4.2,1) {}; 
\node (P411) at (3.2, -1) {};
\node (P412) at (3.7, -1) {};
\node(PK1) at (2,2) {};
\node(PK2) at (2.5, 2) {};
\begin{scope}[on background layer]
    \fill[gray!30, rounded corners] (-0.2,-0.2) rectangle (0.7,0.2);
    \fill[gray!30, rounded corners] (1.3,-0.2) rectangle (2.2,0.2); 
    \fill[gray!30, rounded corners] (0,0.8) rectangle (1.9,1.2); 
    \fill[gray!30, rounded corners] (0.5,-0.8) rectangle (1.4,-1.2); 
    \fill[gray!30, rounded corners] (2.3,-0.2) rectangle (3.2,0.2); 
    \fill[gray!30, rounded corners] (3.8,-0.2) rectangle (4.7,0.2); 
    \fill[gray!30, rounded corners] (2.5,0.8) rectangle (4.4,1.2); 
    \fill[gray!30, rounded corners] (3,-0.8) rectangle (3.9,-1.2); 
    \fill[gray!30, rounded corners] (1.8,1.8) rectangle (2.7,2.2); 
\end{scope}
\node[draw=none, fill=none] at (0.1,-0.4) {$P_{i}'$}; 
\node[draw=none, fill=none] at (1.9,-0.4) {$P_{i}''$}; 
\node[draw=none, fill=none] at (2.5,-0.4) {$P_{j}'$}; 
\node[draw=none, fill=none] at (4.4,-0.4) {$P_{j}''$}; 
\node[draw=none, fill=none] at (2.3,2.4) {$P_{k}$};  
\node[draw=none, fill=none] at (0.5,1.3) {$F$}; 
\node[draw=none, fill=none] at (4,1.3) {$F'$}; 
\node[draw=none, fill=none] at (1,-1.4) {$P_{i}$}; 
\node[draw=none, fill=none] at (3.5,-1.4) {$P_{j}$}; 
\draw (Pi1) -- (Fk1);
\draw (Pi1) -- (Fk2);
\draw (Pi2) -- (Fk3);
\draw (Pi2) -- (Fk4);

\draw (Pj1) -- (Fk1);
\draw (Pj1) -- (Fk3);
\draw (Pj2) -- (Fk2);
\draw (Pj2) -- (Fk4);

\draw (Pi1) -- (P41);
\draw (Pi2) -- (P42);
\draw (Pj1) -- (P41);
\draw (Pj2) -- (P42);

\draw (Pi11) -- (Fk11);
\draw (Pi11) -- (Fk12);
\draw (Pi12) -- (Fk13);
\draw (Pi12) -- (Fk14);

\draw (Pj11) -- (Fk11);
\draw (Pj11) -- (Fk13);
\draw (Pj12) -- (Fk12);
\draw (Pj12) -- (Fk14);

\draw (Pi11) -- (P411);
\draw (Pi12) -- (P412);
\draw (Pj11) -- (P411);
\draw (Pj12) -- (P412);

\draw (PK1) -- (Fk1);
\draw (PK1) -- (Fk4);
\draw (PK2) -- (Fk3);
\draw (PK2) -- (Fk2);

\draw (PK1) -- (Fk11);
\draw (PK1) -- (Fk14);
\draw (PK2) -- (Fk13);
\draw (PK2) -- (Fk12);
\end{tikzpicture}
\caption{An illustration of the gadget (\cite{AKRV17}) to replace $g_{k} = g_{i} \lor g_{j}$ in our construction} 
\label{fig:or-gate}
\end{figure}
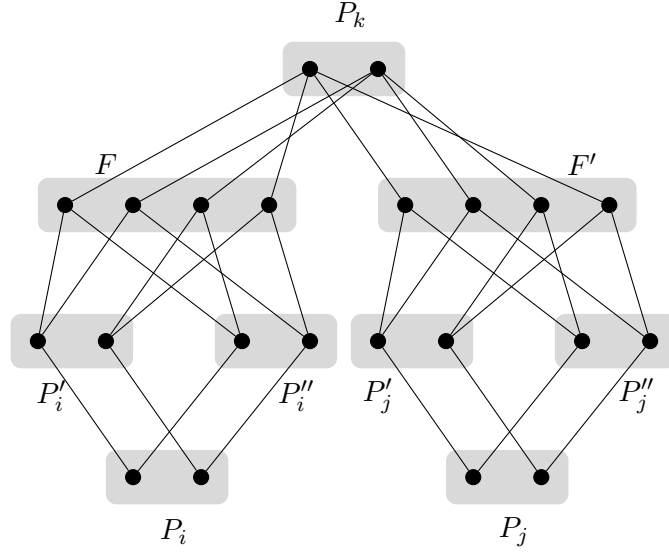

 Next, we modify the graph $G$ as follows:
\begin{itemize}
\item Attach each intermediate output pair of AND/OR gate to the IMP gadget $Y_{k+4}$, where the CFI gadget $X_{3}$ is replaced by $X_{k+4}$. 
\item Connect $P_{k+3}$ and $P_{k+4}$ to a CFI gadget $X_{3} = (P_{k+3}, P_{k+4}, P_{m})$. 
\item Connect $P_{m}$ to all constant $0$ input vertices by two parallel edges.
\end{itemize} 
The resulting graph is $N$. An illustration of the construction for $k = 1$ is shown in~\cref{fig:hardness}. By the properties of CFI gadgets~\cite{CFI92}, the pairs $P_0, P_1, \dots, P_{k+4}, P_{m}$ and the intermediate color classes $F$ and $F'$ all form distinct color classes. We first establish a structural property of $N$, and then use it to prove the correctness of our reduction.

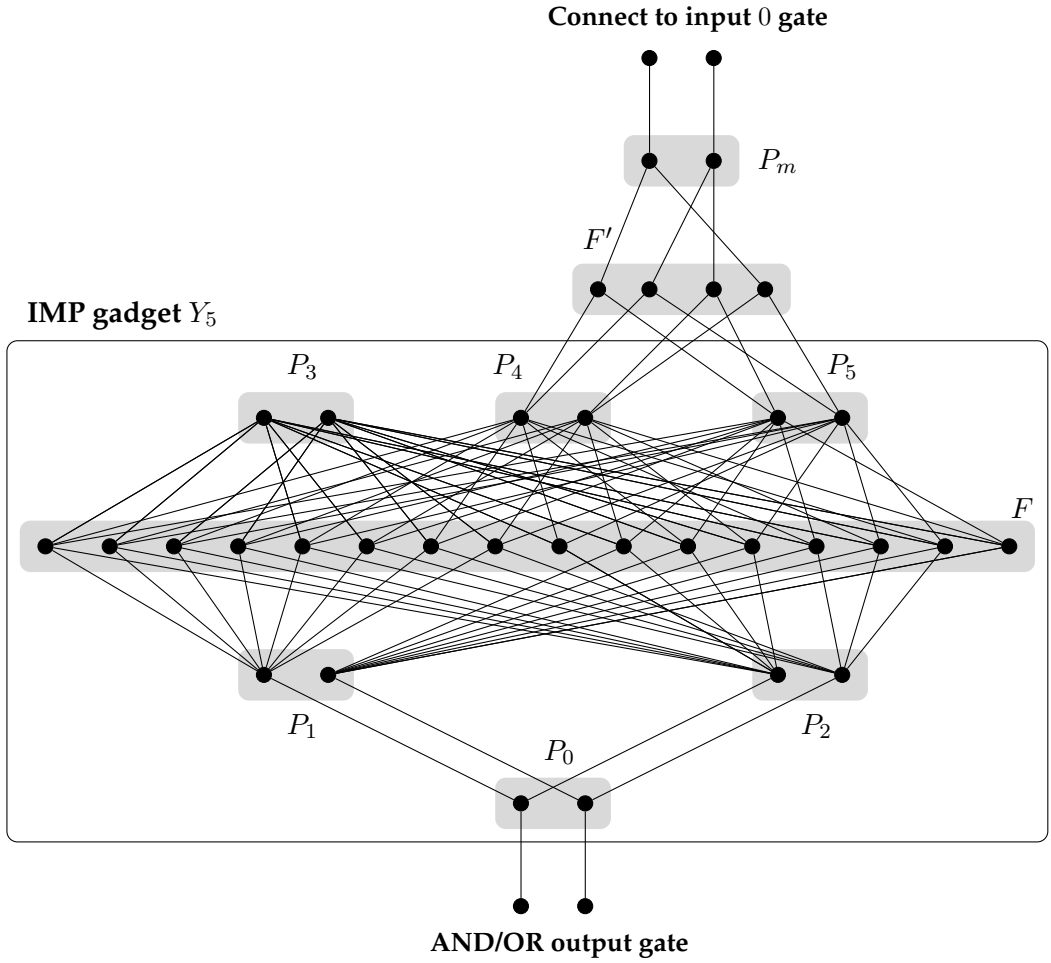
\begin{figure}[!ht] 
\centering
\begin{tikzpicture}[scale=1.7, every node/.style={circle, draw, fill=black, inner sep=2pt}]
\node (Pi1) at (1.5,0) {}; \node (Pi2) at (2,0) {};
\node (Pj1) at (5.5,0) {}; \node (Pj2) at (6,0) {};
\node (Fk1) at (-0.2,1) {}; \node (Fk2) at (0.3,1) {};
\node (Fk3) at (0.8,1) {}; \node (Fk4) at (1.3,1) {};
\node (Fk5) at (1.8,1) {}; \node (Fk6) at (2.3,1) {};
\node (Fk7) at (2.8,1) {}; \node (Fk8) at (3.3,1) {};
\node (Fk9) at (3.8,1) {}; \node (Fk10) at (4.3,1) {};
\node (Fk11) at (4.8,1) {}; \node (Fk12) at (5.3,1) {};
\node (Fk13) at (5.8,1) {}; \node (Fk14) at (6.3,1) {};
\node (Fk15) at (6.8,1) {}; \node (Fk16) at (7.3,1) {};
\node (P31) at (1.5,2) {};  \node (P32) at (2,2) {};
\node (P41) at (3.5,2) {};  \node (P42) at (4,2) {};  
\node (P51) at (5.5,2) {}; \node (P52) at (6,2) {};
\node (F31) at (4.1,3) {}; \node (F32) at (4.5,3) {};
\node (F33) at (5,3) {}; \node (F34) at (5.4,3) {}; 
\node (P01) at (4.5,4) {};
\node (P02) at (5,4) {}; 
\node (Pd1) at (3.5,-1) {}; \node (Pd2) at (4,-1) {};
\node (Po1) at (4.5,4.8) {};   \node (Po2) at (5,4.8) {}; 
\node (PD1) at (3.5,-1.8) {};    \node (PD2) at (4,-1.8) {}; 

\begin{scope}[on background layer]
       \fill[white, draw=black, rounded corners] (-0.5,-1.3) rectangle (7.6,2.6); 
    \fill[gray!30, rounded corners] (1.3,-0.2) rectangle (2.2,0.2);
    \fill[gray!30, rounded corners] (5.3,-0.2) rectangle (6.2,0.2); 
    \fill[gray!30, rounded corners] (-0.4,0.8) rectangle (7.5,1.2);
      \fill[gray!30, rounded corners] (3.3,-1.2) rectangle (4.2,-0.8); 
    \fill[gray!30, rounded corners] (1.3,1.8) rectangle (2.2,2.2); 
    \fill[gray!30, rounded corners] (3.3,1.8) rectangle (4.2,2.2); 
    \fill[gray!30, rounded corners] (5.3,1.8) rectangle (6.2,2.2);
 \fill[gray!30, rounded corners] (3.9,2.8) rectangle (5.6,3.2);
    \fill[gray!30, rounded corners] (4.3,3.8) rectangle (5.2,4.2);
\end{scope}
\node[draw=none, fill=none] at (5.8,-0.4) {$P_2$}; 
\node[draw=none, fill=none] at (1.8,-0.4) {$P_1$}; 
\node[draw=none, fill=none] at (1.8,2.4) {$P_3$}; 
\node[draw=none, fill=none] at (3.4,2.4) {$P_4$}; 
\node[draw=none, fill=none] at (6,2.4) {$P_5$}; 
\node[draw=none, fill=none] at (7.4,1.3) {$F$}; 
\node[draw=none, fill=none] at (3.8,-0.6) {$P_{0}$}; 
\node[draw=none, fill=none] at (5.5,4) {$P_m$};  
\node[draw=none, fill=none] at (4.1,3.4) {$F'$}; 
 \node[draw=none, fill=none] at (0.4,2.8) {\textbf{IMP gadget $Y_{5}$}}; 
\node[draw=none, fill=none, overlay] at (4.8,5.1) {\small{\textbf{Connect to input $0$ gate}}};
\node[draw=none, fill=none, overlay] at (3.8,-2.1) {\small{\textbf{AND/OR output gate}}};
\draw (Pi1) -- (Fk1); \draw (Pi1) -- (Fk2);
\draw (Pi1) -- (Fk3); \draw (Pi1) -- (Fk4);
\draw (Pi1) -- (Fk5); \draw (Pi1) -- (Fk6);
\draw (Pi1) -- (Fk7); \draw (Pi1) -- (Fk8);
\draw (Pj1) -- (Fk9); \draw (Pi2) -- (Fk10);
\draw (Pi2) -- (Fk11); \draw (Pi2) -- (Fk12);
\draw (Pi2) -- (Fk13); \draw (Pi2) -- (Fk14);
\draw (Pi2) -- (Fk15); \draw (Pi2) -- (Fk16);
\draw (Pj1) -- (Fk1); \draw (Pj1) -- (Fk2);
\draw (Pj1) -- (Fk3); \draw (Pj1) -- (Fk4);
\draw (Pj1) -- (Fk9); \draw (Pj1) -- (Fk10);
\draw (Pj1) -- (Fk11); \draw (Pj1) -- (Fk12);
\draw (Pj2) -- (Fk5); \draw (Pj2) -- (Fk6); 
\draw (Pj2) -- (Fk7); \draw (Pj2) -- (Fk8);
\draw (Pj2) -- (Fk13); \draw (Pj2) -- (Fk14);
\draw (Pj2) -- (Fk15); \draw (Pi2) -- (Fk16);
\draw (P31) -- (Fk1); \draw (P31) -- (Fk2);
\draw (P31) -- (Fk5); \draw (P31) -- (Fk6);
\draw (P31) -- (Fk9); \draw (P31) -- (Fk10);
\draw (P31) -- (Fk13); \draw (P31) -- (Fk14);
\draw (P32) -- (Fk3); \draw (P32) -- (Fk4);
\draw (P32) -- (Fk7); \draw (P32) -- (Fk8);
\draw (P32) -- (Fk11); \draw (P32) -- (Fk12);
\draw (P32) -- (Fk15); \draw (P32) -- (Fk16);
\draw (P31) -- (Fk1); \draw (P31) -- (Fk2);
\draw (P31) -- (Fk5); \draw (P31) -- (Fk6);
\draw (P31) -- (Fk9); \draw (P31) -- (Fk10);
\draw (P31) -- (Fk13); \draw (P31) -- (Fk14);
\draw (P32) -- (Fk3); \draw (P32) -- (Fk4);
\draw (P32) -- (Fk7); \draw (P32) -- (Fk8);
\draw (P32) -- (Fk11); \draw (P32) -- (Fk12);
\draw (P32) -- (Fk15); \draw (P32) -- (Fk16);
\draw (P41) -- (Fk1); \draw (P41) -- (Fk3);
\draw (P41) -- (Fk5); \draw (P41) -- (Fk7);
\draw (P41) -- (Fk9); \draw (P41) -- (Fk11);
\draw (P41) -- (Fk13); \draw (P41) -- (Fk15);
\draw (P42) -- (Fk2); \draw (P42) -- (Fk4);
\draw (P42) -- (Fk6); \draw (P42) -- (Fk8);
\draw (P42) -- (Fk10); \draw (P42) -- (Fk12);
\draw (P42) -- (Fk14); \draw (P42) -- (Fk16);
\draw (P51) -- (Fk1); \draw (P51) -- (Fk4);
\draw (P51) -- (Fk6); \draw (P51) -- (Fk7);
\draw (P51) -- (Fk10); \draw (P51) -- (Fk11);
\draw (P51) -- (Fk13); \draw (P51) -- (Fk16);
\draw (P52) -- (Fk2); \draw (P52) -- (Fk3);
\draw (P52) -- (Fk5); \draw (P52) -- (Fk8);
\draw (P52) -- (Fk9); \draw (P52) -- (Fk12);
\draw (P52) -- (Fk14); \draw (P52) -- (Fk15);
\draw (Pd1) -- (Pi1); \draw (Pd2) -- (Pi2);
\draw (Pd1) -- (Pj1); \draw (Pd2) -- (Pj2);
\draw (P41) -- (F31); \draw (P41) -- (F32);
\draw (P42) -- (F33); \draw (P42) -- (F34);
\draw (P51) -- (F31); \draw (P51) -- (F33);
\draw (P52) -- (F32); \draw (P52) -- (F34);
\draw (P01) -- (F31); \draw (P01) -- (F34);
\draw (P02) -- (F32); \draw (P02) -- (F33);
\draw (P01) -- (Po1); \draw (P02) -- (Po2);
\draw (PD1) -- (Pd1); \draw (PD2) -- (Pd2);
\end{tikzpicture}
\vspace{4mm}
\caption{An illustration of the gadget that forms the foundation for constructing the graph that is $1$-Tinhofer, but not $2$-Tinhofer. The graph contains IMP gadget $Y_{5} = (P_{0}, P_1, P_2, P_3, P_4, P_5)$ (for $k = 1$). The last two vertex pairs $P_{4}$ and $P_{5}$ of $Y_{5}$ are connected to CFI graph $X_{3} (P_{m}, P_{4}, P_{5})$.} 
\label{fig:hardness}
\end{figure}

\begin{lemma}
\label{lemma:comp-k-not-k+1-tin}
The vertex pair $P_{m}$ is flipped if and only if the vertex pairs $P_{1}, P_{2}, \dots, P_{k+4}$ are flipped an odd number of times in $N$. 
\end{lemma}
\begin{proof}
If $P_{m}$ is flipped, then $P_{1}$ and $P_{2}$ must be flipped according to the construction. In the gadget $X_3$, one of $P_{k+3}$ or $P_{k+4}$ must be flipped to preserve automorphisms of the CFI gadget. Then, by the parity property of $X_{k+4}$, the pairs $P_{3}, \dots, P_{k+4}$ must be flipped an odd number of times.    
\end{proof}

\begin{claim}
If $C$ evaluates to $1$, then $N$ is discrete.
\label{claim:c-1}
\end{claim}
\begin{proof}
Suppose $C$ evaluates to $1$. Then all gate pairs are refined by color refinement. The additional edges from $P_m$ to each constant-$0$ vertices ensure that these vertices are also refined. Hence all vertices are individualized and $N$ is discrete.
\end{proof}
\begin{claim}
If $C$ evaluates to $0$, then $N$ is $k$-Tinhofer but not ($k+1)$-Tinhofer.
\label{claim:c-0}
\end{claim}
\begin{proof}
Suppose $C$ evaluates to $0$. We consider two identical copies $N'$ and $N''$ of $N$, and run Tinhofer's algorithm on them. We claim that $N$ is $k$-Tinhofer. After any sequence of $k$ individualizations performed on corresponding color classes of $N'$ and $N''$, the CFI gadgets $X_{3}$ and $X_{k+4}$ still satisfy the even-flip parity property. Hence, the resulting colored graphs remain isomorphic after $k$ steps.   
 To show that $N$ is not $(k + 1)$-Tinhofer, we consider a sequence of individualizations for which the resulting colored graphs are not isomorphic after the ($k+1$)-st step. First, we individualize different vertices of $P_{m}$ in $N'$ and $N''$. The refinement step splits $P_{0}, P_{1}, P_{2}$ and $F$ due to the two parallel edges from the final output gate to the constant $0$ input gate. Next, we individualize the same vertex in $P_{3}$ in both graphs $N'$ and $N''$. The refinement step splits only $F$. We continue similarly for the next $(k-1)$ rounds, i.e., we individualize the same vertex in color classes $P_{4}, \dots, P_{k+4}$. 
At this point, the two colored graphs have the vertices of $P_{3}, P_4, \dots, P_{k+4}$ unflipped, while the vertices of $P_{m}$ are flipped. Hence, by~\cref{lemma:comp-k-not-k+1-tin}, these two colored graphs are not isomorphic. But $N'$ and $N''$ are identical copies of $N$, so we have a contradiction. Thus, $N$ is not $(k + 1)$-Tinhofer. 
\end{proof}

\begin{lemma}
There is a uniform $\ACz$ many-one reduction from \MCVP \, to \ktintest.
\label{lemma:hard}
\end{lemma}
\begin{proof}
 Given a monotone Boolean circuit $C$, we construct the graph $N$ as described above. By \cref{claim:c-1}, if $C$ evaluates to $1$, then $N$ is discrete. By \cref{claim:c-0}, if $C$ evaluates to $0$, then $N$ is $k$-Tinhofer but not $(k + 1)$-Tinhofer. Hence, this gives a reduction from \MCVP \, to \ktintest. Finally, the construction of $N$ from $C$ is uniform and local. Each gate is replaced by a constant-size gadget and edges are added according to the wiring of $C$. Hence, the reduction is computable by uniform $\ACz$ circuits.
\end{proof}

\section{Fixed-Parameter Tractability}
\label{sec:fpt}
We now present an algorithm to test isomorphism within the lower levels of the hierarchy (\cref{fig:hierarchy}) in fixed-parameter tractable time parameterized by a new parameter \emph{Tinhofer deficiency}. The \emph{Tinhofer deficiency} of a graph $G$ on $n$ vertices is defined to be the largest integer $k$ for which $G$ is not $(n - k)$-Tinhofer. We present the proof of~\cref{prop:introfptalgoforktin}.

\begin{proof}
Suppose $G$ has Tinhofer deficiency $k$ and $H$ is any graph. We run Tinhofer's algorithm on $G$ and $H$ for $(n-k)$ rounds of individualization and refinement. For the remaining $k$ vertices, we consider all possible mappings between them that respect the colors. If any such mapping corresponds to an isomorphism between $G$ and $H$, then the algorithm outputs $G \cong H$. Otherwise, if none of the mappings yields an isomorphism, then the algorithm outputs $G \ncong H$. Since $G$ has Tinhofer deficiency $k$, it is $(n-k)$-Tinhofer. So, Tinhofer's algorithm does not make an error when ($n-k)$ vertices are individualized. Thus, if $G \cong H$, then there exists a mapping between the remaining $k$ vertices. The total number of possible mappings is $k! \le k^{k} = 2^{k \log k}$. The running time of the algorithm is $2^{k \log k} \cdot n^{\mathcal{O}(1)}$, which is fixed-parameter tractable with respect to the Tinhofer deficiency $k$. If $k$ is constant, then this algorithm runs in polynomial time. 
\end{proof}

\section{Discussions and Conclusion}
\label{sec:conclusion}

In this paper, we studied a new hierarchy of graphs within the class of Tinhofer graphs. We proved that the hierarchy is strict by constructing explicit graph examples separating successive levels of the hierarchy, and we also proved that testing membership in the $(k+1)$-th level of the hierarchy for graphs promised to be in the $k$-th level is $\P$-hard via a reduction from the monotone circuit value problem. In this context, it would be interesting to find more natural families of graphs that separate the levels of the $k$-Tinhofer hierarchy.

We also presented a fixed-parameter tractable algorithm for testing isomorphism within the lower levels of the hierarchy with respect to Tinhofer deficiency as the parameter. It is unclear whether this problem is also kernelizable with respect to Tinhofer deficiency as a parameter. 

Indeed, understanding the exact complexity of the \ktintest problem is also challenging. Note that the problem is at least as hard as isomorphism testing of vertex-transitive graphs when $k=0$ \cite{AKRV17}. It is conceivable that for other values of $k$, the problem is at least as hard as some variants of the graph isomorphism problem. 
\bibliographystyle{alpha}
\bibliography{reference-journal}

@article{CFI92,
  title={An optimal lower bound on the number of variables for graph identifications},
  author={Cai, Jin-Yi and F{\"u}rer, Martin and Immerman, Neil},
  journal={Combinatorica},
  volume={12},
  number={4},
  pages={389--410},
  year={1992}
}

@article{AKRV17,
  title={Graph isomorphism, color refinement, and compactness},
  author={Arvind, Vikraman and K{\"o}bler, Johannes and Rattan, Gaurav and Verbitsky, Oleg},
  journal={Computational Complexity},
  volume={26},
  number={3},
  pages={627--685},
  year={2017}
}

@article{KS09,
  title={Equitable partitions and orbit partitions},
  author={Kudose, Satoru},
  journal={Acta Mathematica Sinica},
  pages={1--9},
  year={2009},
}

@article{F49,
  title={Graphs of degree three with a given abstract group},
  author={Frucht, Robert},
  journal={Canadian Journal of Mathematics},
  volume={1},
  number={4},
  pages={365--378},
  year={1949}
}

@article{KSS21,
  title={Graphs identified by logics with counting},
  author={Kiefer, Sandra and Schweitzer, Pascal and Selman, Erkal},
  journal={ACM Transactions on Computational Logic (TOCL)},
  volume={23},
  number={1},
  pages={1--31},
  year={2021},
  publisher={ACM New York, NY}
}

@inproceedings{Bab16,
  title={Graph isomorphism in quasi-polynomial time},
  author={Babai, L{\'a}szl{\'o}},
  booktitle={Proceedings of the 48th Annual ACM Symposium on Theory of Computing (STOC)},
  pages={684--697},
  year={2016}
}

@article{DGR18,
  title={Lov{\'a}sz Meets {W}eisfeiler and {L}eman},
  author={Dell, Holger and Grohe, Martin and Rattan, Gaurav},
  journal={arXiv preprint arXiv:1802.08876},
  year={2018}
}

@article{Luk82,
  title={Isomorphism of graphs of bounded valence can be tested in polynomial time},
  author={Luks, Eugene M},
  journal={Journal of computer and system sciences},
  volume={25},
  number={1},
  pages={42--65},
  year={1982}
}

@article{WL68,
  title={The reduction of a graph to canonical form and the algebra which appears therein},
  author={Weisfeiler, Boris and Leman, Andrei},
  journal={Nauchno-Technicheskaya Informatsia (NTI) Series},
  volume={2},
  number={9},
  pages={12--16},
  year={1968}
}

@inproceedings{GKMS14,
  title={Dimension reduction via colour refinement},
  author={Grohe, Martin and Kersting, Kristian and Mladenov, Martin and Selman, Erkal},
  booktitle={Proceedings of the 22th Annual European Symposium on Algorithms (ESA)},
  pages={505--516},
  year={2014}
}

@inproceedings{KMGG14,
  title={Power iterated color refinement},
  author={Kersting, Kristian and Mladenov, Martin and Garnett, Roman and Grohe, Martin},
  booktitle={Proceedings of the 28th AAAI Conference on Artificial Intelligence (AAAI)},
pages={1904--1910},
year={2014}
}

@inproceedings{Bab82,
  title={Isomorphism of graphs with bounded eigenvalue multiplicity},
  author={Babai, L{\'a}szl{\'o} and Grigoryev, D Yu and Mount, David M},
  booktitle={Proceedings of the 14th Annual ACM Symposium on Theory of Computing (STOC)},
  pages={310--324},
  year={1982}
}

@inproceedings{HW74,
  title={Linear time algorithm for isomorphism of planar graphs (preliminary report)},
  author={Hopcroft, John E and Wong, Jin-Kue},
  booktitle={Proceedings of the 6th Annual ACM Symposium on Theory of Computing (STOC)},
  pages={172--184},
  year={1974}
}

@article{Bod90,
  title={Polynomial algorithms for graph isomorphism and chromatic index on partial k-trees},
  author={Bodlaender, Hans L},
  journal={Journal of Algorithms},
  volume={11},
  number={4},
  pages={631--643},
  year={1990}
}

@article{Tin91,
  title={A note on compact graphs},
  author={Tinhofer, Gottfried},
  journal={Discrete Applied Mathematics},
  volume={30},
  number={2-3},
  pages={253--264},
  year={1991}
}

@article{Go77,
  title={The monotone and planar circuit value problems are log space complete for {P}},
  author={Goldschlager, Leslie M},
  journal={ACM {SIGACT} {N}ews},
  volume={9},
  number={2},
  pages={25--29},
  year={1977}
}

@phdthesis{Ki20,
  title={Power and limits of the Weisfeiler-Leman algorithm},
  author={Kiefer, Sandra},
  year={2020},
  school={Dissertation, RWTH Aachen University, 2020}
}

@inproceedings{GKMS21,
  title={Color refinement and its applications},
  author={Grohe, Martin and Kersting, Kristian and Mladenov, Martin and Schweitzer, Pascal},
 booktitle={An Introduction to Lifted Probabilistic Inference},
  year={2021},
  publisher={MIT Press}
}

@book{IL90,
  title={Describing graphs: A first-order approach to graph canonization},
  author={Immerman, Neil and Lander, Eric},
  year={1990},
  publisher={Springer}
}

@inproceedings{BG15,
  title={Limitations of algebraic approaches to graph isomorphism testing},
  author={Berkholz, Christoph and Grohe, Martin},
  booktitle={International Colloquium on Automata, Languages, and Programming (ICALP)},
  pages={155--166},
  year={2015}
}

@article{K20,
  title={The {W}eisfeiler-{L}eman algorithm: an exploration of its power},
  author={Kiefer, Sandra},
  journal={ACM SIGLOG News},
  volume={7},
  number={3},
  pages={5--27},
  year={2020}
}

@article{Dvo10,
  title={On recognizing graphs by numbers of homomorphisms},
  author={Dvo{\v{r}}{\'a}k, Zden{\v{e}}k},
  journal={Journal of Graph Theory},
  volume={64},
  number={4},
  pages={330--342},
  year={2010}
}

@inproceedings{KV15,
  title={Universal covers, color refinement, and two-variable counting logic: Lower bounds for the depth},
  author={Krebs, Andreas and Verbitsky, Oleg},
  booktitle={30th Annual ACM/IEEE Symposium on Logic in Computer Science (LICS)},
  pages={689--700},
  year={2015}
}

@article{MB81,
  title={Practical graph isomorphism},
  author={McKay, Brendan D},
  journal={Congressus Numerantium},
  volume={30},
  pages={45-87},
  year={1981}
}

@inproceedings{MRFL19,
  title={Weisfeiler and {L}eman go neural: Higher-order graph neural networks},
  author={Morris, Christopher and Ritzert, Martin and Fey, Matthias and Hamilton, William L and Lenssen, Jan Eric and Rattan, Gaurav and Grohe, Martin},
  booktitle={33rd {AAAI} Conference on Artificial Intelligence (AAAI)},
  volume={33},
  pages={4602--4609},
  year={2019}
}

@inproceedings{ABS21,
  title={A Characterization of Individualization-Refinement Trees},
  author={Anders, Markus and Brachter, Jendrik and Schweitzer, Pascal},
  booktitle={32nd International Symposium on Algorithms and Computation (ISAAC)},
  pages={24--1},
  year={2021}
}

@article{MLM23,
  title={Weisfeiler and {L}eman go machine learning: The story so far},
  author={Morris, Christopher and Lipman, Yaron and Maron, Haggai and Rieck, Bastian and Kriege, Nils M and Grohe, Martin and Fey, Matthias and Borgwardt, Karsten},
  journal={Journal of Machine Learning Research},
  volume={24},
  number={333},
  pages={1--59},
  year={2023}
}

@article{AFKKR22,
  title={The parameterized complexity of fixing number and vertex individualization in graphs},
  author={Arvind, Vikraman and Fuhlbrueck, Frank and Koebler, Johannes and Kuhnert, Sebastian and Rattan, Gaurav},
  journal={ACM Transactions on Computation Theory (TOCT)},
  volume={14},
  number={2},
  pages={1--26},
  year={2022}
}

@article{HHBD17,
  title={Graph isomorphisms in quasi-polynomial time},
  author={Helfgott, Harald Andr{\'e}s and Bajpai, Jitendra and Dona, Daniele},
  journal={arXiv preprint arXiv:1710.04574},
  year={2017}
}

@article{DH17,
  title={Pebble games with algebraic rules},
  author={Dawar, Anuj and Holm, Bjarki},
  journal={Fundamenta Informaticae},
  volume={150},
  number={3-4},
  pages={281--316},
  year={2017},
  publisher={SAGE Publications Sage UK: London, England}
}

\end{document}